\documentclass[twocolumn]{autart} 

\usepackage{amsmath}
\usepackage{array}
\usepackage{cite}
\usepackage{amsfonts}
\usepackage{algorithm}
\usepackage{graphicx} % for pdf, bitmapped graphics files
\usepackage{epstopdf}
\usepackage{float}
\usepackage{balance}
\usepackage{enumitem}
\usepackage[fancythm,fancybb,morse,elsart]{jphmacros2e} 
\usepackage{mathrsfs}
\usepackage{centernot}
\usepackage[utf8]{inputenc}
\usepackage[english]{babel}
\usepackage{tikz}
\usetikzlibrary{shapes,arrows}
\usepackage{bbm}
\usepackage{algpseudocode}
\usepackage{siunitx}
\usepackage{stackengine}
\usepackage{mathtools}
\usepackage{amssymb}
\usepackage{tabularx}
\usepackage{collcell}
\usepackage{subcaption}
\usepackage{caption}

\usepackage{balance}
\usepackage{comment}
\usepackage{etoolbox}
\usepackage{booktabs}
\usepackage{titlesec}

\newlength{\mywidth}
\newcommand{\Rcolumn}[1]{\mathmakebox[\mywidth][c]{#1}}
\newcolumntype{R}{>{\collectcell\Rcolumn}c<{\endcollectcell}}

\newsavebox{\tempbox}

 \usepackage[norefs,nocites]{refcheck}
\allowdisplaybreaks

\newcommand{\norm}[1]{\left\lVert#1\right\rVert}

\begin{document}
\begin{frontmatter}

\title{Input-Output Data-Driven Sensor Selection for Cyber-Physical Systems\thanksref{footnoteinfo}}

\thanks[footnoteinfo]{This work was supported in part by ARO under grant No. W$911$NF-$24-1-0174$, and by NSF under grant Nos. CPS-$2227185$, SATC-$2231651$, SLES-$2415479$.}

\author[1]{Filippos Fotiadis}\ead{ffotiadis@utexas.edu},    % Add the 
\author[2]{Kyriakos G. Vamvoudakis}\ead{kyriakos@gatech.edu}  % (ead) as shown

\address[1]{The Oden Institute for Computational Engineering \& Sciences, The University of Texas at Austin, Austin, TX, USA}\vspace{-2mm}
\address[2]{The Daniel Guggenheim School of Aerospace Engineering, 
Georgia Institute of Technology, Atlanta, GA, 
USA}  % Please supply   


\begin{abstract}
In this paper, we consider the problem of input-output data-driven sensor selection for unknown cyber-physical systems (CPS). In particular, out of a large set of sensors available for use, we choose a subset of them that maximizes a metric of observability of the CPS. The considered observability metric is related to the $\mathcal{H}_2$ system norm, which quantifies the average output energy of the selected sensors over a finite or an infinite horizon. However, its computation inherently requires knowledge of the unknown matrices of the system, so we draw connections from the reinforcement learning literature and design an input-output data-driven algorithm to compute it in a model-free manner. We then use the derived data-driven metric expression to choose the best sensors of the system in polynomial time, effectively obtaining a provably convergent model-free sensor selection process. Additionally, we show how the proposed data-driven approach can be exploited to select sensors that optimize volumetric measures of observability, while also noting its applicability to the dual problem of actuator selection. Simulations are performed to demonstrate the validity and effectiveness of the proposed approach.
\end{abstract}

\end{frontmatter}

\section{Introduction}

Cyber-physical systems (CPS) are heterogeneous systems that integrate analog and digital components, along with communication channels through which these exchange data. Among the most critical components of a CPS are its sensors and actuators, which must be carefully selected to ensure sufficient observability and controllability of the system \cite{pasqualetti2014controllability, liu2013observability}. This problem of properly choosing the CPS' sensors (or actuators) is known as the sensor (or actuator) selection problem \cite{Taha, manohar2021optimal}.

The integration of heterogeneous components makes CPS widely applicable, but also challenging to control and prone to faults that can make prior models obsolete. This has led to growing interest in data-driven control strategies that rely solely on measured data, bypassing explicit models \cite{lewis2009reinforcement}. However, while data-driven control has received considerable attention, sensor and actuator selection in this setting remains largely unexplored. Most learning-based methods in this domain are heuristic, lack suitable theoretical guarantees, and are based mainly on intuition and engineering \cite{wang2019reinforcement, guo2008adaptive, semaan2017optimal}. In addition, while learning-based methods for actuator selection were recently presented in \cite{ye2024online, fotiadis2023data, fotiadis2024learning, azuma},  those rely on the assumption of full-state feedback and do not tackle the problem of sensor selection. This gap in the literature motivates the present work: we present the first input-output data-driven sensor selection procedure for CPS that is agnostic to system models and which comes with theoretical guarantees of convergence.

\textit{Related Work:} In most cases, sensor selection concerns choosing sensors that maximize a metric of observability \cite{wolf2012optimizing}. For example, \cite{manohar2021optimal, clark2020sensor, chen2011h2, summers2014optimal} developed algorithms for the selection of sensors that optimize $\mathcal{H}_2$ norm-related metrics of the system. The main idea was that efficient sensors should generally yield stronger signals, and this strength is captured by the $\mathcal{H}_2$ norm. Other studies for observability-based sensor selection also regard the optimization of Gramian-related metrics \cite{summers2016convex, yamada2023efficient, bopardikar2017sensor, siami2020deterministic}. However, non-observability-based sensor selection was also proposed in \cite{dhingra2014admm, zare2019proximal, tzoumas_lqg} for linear quadratic regulation, and in \cite{milovsevic2020actuator, pirani2021game} for security.

A limitation of the above results is that they are model-based, requiring complete knowledge of the system dynamics to perform the sensor selection procedure. This can often be a restrictive requirement for CPS, which are complex by definition and thus often contain unmodeled dynamics. On the other hand, existing learning-based methods mainly concern the problem of actuator selection and assume full-state feedback \cite{ye2024online, fotiadis2023data, fotiadis2024learning}, are heuristic or concern static system models \cite{wang2019reinforcement, guo2008adaptive, semaan2017optimal}, or use learning for purposes unrelated to having unknown system dynamics \cite{manohar2018data, vafaee2024learning, ballota}. These limitations create motivation for the development of model-free sensor selection procedures with theoretical guarantees of convergence.

RL and adaptive control are frequently employed to tackle unknown dynamics in control \cite{lewis2009reinforcement}. For example, \cite{kiumarsi2017optimal} surveys learning-based adaptive control procedures for the optimal control of unknown systems using input-state data, and a framework solely using input-output data was presented in \cite{lewis2010reinforcement}. The main advantage of these studies is that they are not based on heuristic arguments, however, their main aim is to design \textit{controllers}. On the other hand, the main goal of the present work is to design data-driven algorithms for \textit{selection}. 

\textit{Contributions}: Inspired by the aforementioned RL-based studies for optimal control, we design the first model-free sensor selection algorithm with theoretical guarantees of convergence. More specifically, we consider the problem of choosing sensors that maximize observability metrics related to the $\mathcal{H}_2$ norm of the system, which requires solving a set of Lyapunov equations. However, these Lyapunov equations are model-based, so we re-express them solely with respect to input-output data gathered from the system and use the resulting expression to perform guaranteed model-free sensor selection. The proposed process is also applicable to the problem of actuator selection, as well as to the problem of optimizing volumetric observability measures.

A preliminary version appeared in \cite{fotiadis2024input}. This paper extends \cite{fotiadis2024input} by providing full technical proofs, relaxing assumptions on system knowledge and the required rank of the input-output data, addressing finite-horizon and volumetric observability metrics, and discussing observability verification.

\textit{Notation}:  For any signal \( x:\mathbb{N} \to \mathbb{R}^n \), we define \( x(t_1:t_2) = [x^\mathrm{T}(t_2)~x^\mathrm{T}(t_2-1)~\ldots~x^\mathrm{T}(t_1)]^\mathrm{T} \) for \( t_1, t_2 \in \mathbb{N} \) with \( t_2 \ge t_1 \). The norm \( \|\cdot\| \) denotes the Euclidean norm, while \( \rho(\cdot) \), \( \mathrm{tr}(\cdot) \), and \( \det(\cdot) \) denote the spectral radius, trace, and determinant of a square matrix, respectively. For a matrix \( Z \), \( Z^\dagger \) denotes the Moore–Penrose pseudoinverse. If \( Z \in \mathbb{R}^{n \times m} \), then \( \mathrm{vec}(Z) = [z_{1,1}~z_{2,1}~\ldots~z_{n,1}~z_{1,2}~\ldots~z_{n,m}]^\mathrm{T} \) is its vectorized form, and \( \mathrm{vec}^{-1} \) denotes the inverse operation. If \( Z \) is square and symmetric, we define \( \mathrm{vech}(Z) = [z_{1,1}~z_{2,1}~\ldots~z_{n,1}~z_{2,2}~\ldots~z_{n,n}]^\mathrm{T} \) as the half-vectorized form, \( \mathrm{vecs}(Z) = [z_{1,1}~2z_{2,1}~\ldots~2z_{n,1}~z_{2,2}~2z_{3,2}~\ldots~z_{n,n}]^\mathrm{T} \) as the scaled half-vectorized form, and \( \lambda_{\min}(Z) \) as its minimum eigenvalue. Using these, we define the operator \( H(\cdot) = \mathrm{vecs}(\mathrm{vec}^{-1}(\cdot)) \) and its inverse \( H^{-1}(\cdot) = \mathrm{vec}(\mathrm{vecs}^{-1}(\cdot)) \). The symbol \( \otimes \) denotes the Kronecker product, \( e_i \) denotes the \( i \)-th standard basis vector, and \( I_n \) the \( n \times n \) identity matrix.

\section{Problem Formulation}\label{sec:pr}

\subsection{System and Preliminaries}

Consider, for $t\in\mathbb{N}$, a discrete-time system of the form:
\begin{equation}\label{eq:sys}
\begin{split}
x(t+1)&=Ax(t)+Bu(t),~x(0)=x_0,\\
y(t)&=Cx(t),
\end{split}
\end{equation}
where $x(t)\in\mathbb{R}^n$ is the state with initial value $x_0\in\mathbb{R}^n$, $u(t)\in\mathbb{R}^m$ is the control input, $y(t)$ is the output, and $A\in\mathbb{R}^{n\times n}$, $B\in\mathbb{R}^{n\times m}$, $C\in\mathbb{R}^{p\times n}$ are the system matrices. 
Each row of $C$ represents a sensor, whereas each column of $B$ represents an actuator. In particular,
\begin{align*}
C&=\begin{bmatrix} c_1^\mathrm{T} & c_2^\mathrm{T} & \ldots & c_p^\mathrm{T} \end{bmatrix}^\mathrm{T},~B=\begin{bmatrix} b_1 & b_2 & \ldots & b_m \end{bmatrix},
\end{align*}
where $c_j\in\mathbb{R}^{1\times n}$, $j=1,\ldots,p$ represents a sensor, whereas $b_i\in\mathbb{R}^n$, $i=1,\ldots,m$, represents an actuator.

In the sensor selection problem, we assume that we want to select only $p'<p$ out of the $p$ sensors available for use at each time instant. That is, we seek a selection matrix
\begin{equation}\label{eq:Smatrix}
\mathbb{S}_\gamma=\begin{bmatrix}e_{\gamma_1} & e_{\gamma_2} & \ldots & e_{\gamma_{p'}} \end{bmatrix}^\mathrm{T},
\end{equation}
with $\gamma=\{\gamma_1,\ldots,\gamma_{p'}\}\subset\{1,\ldots p\}$, that will choose certain rows of $C$, so that the new sensing matrix is 
\begin{equation*}
\tilde{C}=\mathbb{S}_{\gamma}C.
\end{equation*}
Accordingly, in the actuator selection problem, we want to select only $m'<m$ out of the $m$ actuators available at each time instant through the selection matrix
\begin{equation*}
\mathbb{S}_\beta=\begin{bmatrix}e_{\beta_1} & e_{\beta_2} & \ldots & e_{\beta_{m'}} \end{bmatrix},
\end{equation*}
where $\beta=\{\beta_1,\ldots,\beta_{m'}\}\subset\{1,\ldots m\}$. The new actuating matrix will then be 
\begin{equation*}
\tilde{B}=B\mathbb{S}_\beta.
\end{equation*}
Various studies have been carried out in the literature to solve sensor or actuator selection problems with performance guarantees \cite{summers2014optimal, summers2016convex, yamada2023efficient, bopardikar2017sensor}.
In this work, however, we consider these problems under a setup in which
\begin{enumerate}
\item All system matrices $A, B, C$ are unknown.
\item Only the output provided by the selected sensors at the instant $t\in\mathbb{N}$ is available for measurement, i.e., $y'(t)=\mathbb{S}_\gamma Cx(t)$.
\end{enumerate}
 As a minimum requirement, we still require the following observability Assumption, which implies we know a set of sensors that renders the system observable.

\begin{assumption}\label{ass:obs}
There exists a known selection index $\hat{\gamma}=\{\hat{\gamma}_1,\ldots,\hat{\gamma}_{r}\}\subset \{1,\ldots,p\}$, $r<p'$, such that $(A,~\mathbb{S}_{\hat{\gamma}}C)$ is observable. \frqed
\end{assumption}

\begin{remark}
Assumption \ref{ass:obs} implies we have some empirical knowledge about the system, particularly regarding its observability.  This is often the case, for example, in mechanical systems, where it is well-known that measuring positions or angles makes the closed-loop observable. 
\end{remark}

\begin{remark}
If one has sufficient resources to concurrently utilize all sensors during the data gathering process, then   $\hat{\gamma}$ in Assumption $1$ shall be unknown. \frqed
\end{remark}

\subsection{Observability and Controllability Indices}\label{subsec:gramians}

The next step is to define the cost function that will distinguish which set of sensors or actuators is the best to select, and metrics of observability and controllability become relevant here. Nevertheless, since all system matrices $A, B, C$ are unknown, there is no unique state-space representation for the input-output behavior of \eqref{eq:sys}. This limitation restricts us to metrics that apply to the transfer function $G$ of the tuple $(A,B,C)$, instead of directly to the tuple $(A,B,C)$.

\subsubsection{Infinite-Horizon Indices}

One metric that quantifies controllability and observability is the $\mathcal{H}_2$ norm of the transfer function $G$ \cite{manohar2021optimal}. This is defined as the infinite-horizon average output energy $\norm{G}_2^2=\mathbb{E}[\sum_{t=0}^\infty \norm{y(t)}^2]$
under the condition that $x_0=0$ and $u(0)\sim \mathcal{N}(0,I_m)$.  The motivation behind choosing this norm as the cost function for the sensor selection problem is that more observable sensors will generally yield larger output signals and, similarly, more controllable actuators will have a larger effect on the measured state under the same actuation signal \cite{manohar2021optimal}.  However, a limitation is that it applies only to asymptotically stable systems, owing to the infinite summation being divergent if $A$ is unstable. To circumvent this issue, following \cite{silva2017stochastic} one can discount the output measurements in the $\mathcal{H}_2$ norm with a parameter $a\in(0,1)$, and instead define a discounted $\mathcal{H}_2$ norm as:
\begin{equation}\label{eq:H2a_int}
\norm{G}_{2,a}^2=\mathbb{E}\left[\sum_{t=0}^\infty a^{2(t-1)}\norm{y(t)}^2\right].
\end{equation}
If $a<\frac{1}{\rho(A)}$, similar to \cite{silva2017stochastic}, it is straightforward to verify that \eqref{eq:H2a_int} is well-defined. In particular, it follows that $\norm{G}_{2,a}^2=\|\tilde{G}\|^2_2$ where $\tilde{G}$ is the transfer function of the tuple $(aA,B,C)$. In addition,  it follows from standard linear systems theory \cite{bernstein2009matrix}  that \eqref{eq:H2a_int} is equivalent to
\begin{equation}\label{eq:trinf}
\textrm{tr}(CW_cC^\mathrm{T}),~\textrm{tr}(B^\mathrm{T}W_oB)
\end{equation}
where $W_o,W_c$ are observability and controllability Gramians for the pairs $(aA,C)$ and $(aA,B)$.

\subsubsection{Finite-Horizon Indices}

Infinite-horizon metrics like \eqref{eq:H2a_int} have the advantage that they can be computed with relatively low computations by solving time-independent ALEs. However, such metrics may not be applicable in cases where the system is deployed only for a fixed amount of time, or when a sufficiently small discount factor $a$ is not a priori known. In these setups, one may instead consider the finite-horizon counterpart of the average output energy \eqref{eq:H2a_int}:
\begin{equation}\label{eq:H2a_intT}
\norm{G_T}_2^2=\mathbb{E}\left[\sum_{t=0}^T \norm{y(t)}^2\right],~T\in\mathbb{N},
\end{equation}
which, following standard systems theory, boils down to either of the following two equivalent expressions:
\begin{equation}\label{eq:trfin}
\textrm{tr}(CW_c(T)C^\mathrm{T}),~\textrm{tr}(B^\mathrm{T}W_o(T)B).
\end{equation}
Unlike \eqref{eq:trinf}, the matrices $W_o(T), W_c(T)$ are the \textit{finite-horizon} observability and controllability Gramians of $(A,C)$ and $(A,B)$, whose existence is not conditional on knowing a sufficiently small discount factor $a$.

\begin{comment}
In the same manner, one can obtain the following (equivalent) finite-horizon logarithmic measures of observability and controllability:
\begin{equation}\label{eq:detfin}
\textrm{logdet}(CW_c(T)C^\mathrm{T}),~\textrm{logdet}(B^\mathrm{T}W_o(T)B).
\end{equation}
\end{comment}

\subsection{Data-Driven Sensor and Actuator Selection}

In the rest of this paper, we focus on finding the sensors (or actuators) that maximize 
\eqref{eq:trinf} or \eqref{eq:trfin}. However, in doing so we will rely solely on input-output data, and not on the unknown system matrices $A, B, C$. In particular, for the sensor selection problem and depending on the metric of interest, we want to find the optimal selection matrix $\mathbb{S}_{\gamma^\star}\subseteq S_{p,p'}$, where $S_{p,p'}$ is the set of selection matrices $\mathbb{S}_{\gamma}$ of the form 
 \eqref{eq:Smatrix} in which $\gamma=\{\gamma_1,\ldots,\gamma_{p'}\}\subset\{1,\ldots p\}$, that optimizes either one of the following two measures (as defined previously).

\begin{enumerate}
\item Infinite-horizon average output energy:
\begin{equation}\label{eq:sensor_problem_trinf}
\mathbb{S}_{\gamma^\star}=\argmax_{\mathbb{S}_{\gamma}\subseteq S_{p,p'}} \textrm{tr}(\mathbb{S}_{\gamma}CW_cC^{\mathrm{T}}\mathbb{S}_{\gamma}^\mathrm{T}).
\end{equation}
\item Finite-horizon average output energy:
\begin{equation}\label{eq:sensor_problem_trfin}
\mathbb{S}_{\gamma^\star}=\argmax_{\mathbb{S}_{\gamma}\subseteq S_{p,p'}} \textrm{tr}(\mathbb{S}_{\gamma}CW_c(T)C^{\mathrm{T}}\mathbb{S}_{\gamma}^\mathrm{T}).
\end{equation}
\end{enumerate}

\begin{remark}
The problems of sensor and actuator selection are dual in nature. For this reason, in what follows, we focus only on the problem of sensor selection, while noting that the approach we present applies to the actuator selection problem with minor modifications. \frqed
\end{remark}

\begin{remark}
The $\mathcal{H}_2$ norm is a desirable metric for sensor selection as it reflects the trace of a weighted observability Gramian \cite{summers2014optimal, manohar2021optimal}. For actuator selection, maximizing the $\mathcal{H}_2$ norm also improves controllability as it captures the trace of a weighted controllability Gramian. However, in the presence of input disturbances, additional considerations may be needed to mitigate their influence. \frqed
\end{remark}

\begin{remark}
Instead of maximizing the $\mathcal{H}_2$ norm, one can select sensors by optimizing volumetric observability measures \cite{summers2015submodularity, manohar2021optimal}. We discuss this in Section \ref{sec:log}. \frqed
\end{remark}

\section{Data-Driven Estimation of the Cost Functions for Fixed Sensor Sets}\label{sec:data}

The first step towards solving the sensor selection problem with unknown system matrices $A, B, C$, is being able to express the cost functions in \eqref{eq:sensor_problem_trinf}-\eqref{eq:sensor_problem_trfin} using input-output data only. In this section, we perform this task for a fixed choice of the selection matrix $\mathbb{S}_{\gamma}$, and we will subsequently use this in the next section to perform data-driven sensor selection.

Before proceeding, note that this task requires i) being able to take measurements from the sensors involved in the matrix $\tilde{C}=\mathbb{S}_{\gamma}C$ (since $C$ is unknown), whose cost function we want to evaluate; and ii) concurrently being able to take measurements from the sensors involved in $\hat{C}=\mathbb{S}_{\hat{\gamma}}C$, as a minimum observability requirement. This reduces the output measurements we can access to:
\begin{equation}\label{eq:sys2}
\begin{split}
x(t+1)&=Ax(t)+Bu(t),~x(0)=x_0,\\
\hat{y}(t)&=\hat{C}x(t),~
\tilde{y}(t)=\tilde{C}x(t).
\end{split}
\end{equation}

\subsection{Expressing the State with Input-Output Data}

Apart from the matrices $A, B, C$ being unknown, a major restriction of the present setup is not having access to full-state feedback, rather, only to the measured outputs available at each time. Nevertheless, directly following \cite{lewis2010reinforcement} the state $x(t)$ can be reconstructed using a history stack of past outputs and control inputs. 

\begin{lemma}\label{le:output}\cite{lewis2010reinforcement}
Let Assumption \ref{ass:obs} hold, so that $(A,\hat{C})$ is observable. Define the matrices
\begin{align*}
U_N&=\begin{bmatrix}B & AB & A^2B & \ldots & A^{N-1}B\end{bmatrix},\\
V_N&=\begin{bmatrix}(\hat{C}A^{N-1})^\mathrm{T} & \ldots & (\hat{C}A)^\mathrm{T} & \hat{C}^\mathrm{T} \end{bmatrix}^\mathrm{T},\\
T_N&=\begin{bmatrix}0 & \hat{C}B & \hat{C}AB & \ldots & \hat{C}A^{N-2}B \\ 0 & 0 & \hat{C}B & \ldots & \hat{C}A^{N-3}B \\ \vdots & \vdots & \ddots & \ddots & \vdots \\ 0 & \ldots & \ & 0 & \hat{C}B \\ 0 & 0 & 0 & 0 & 0 \end{bmatrix}.
\end{align*}
If $N\ge K$, where $K$ is the observability index\footnote{The observability index is upper-bounded as $K\le n$.}, then
\begin{equation*}
x(t)=Mz(t),~t\ge N,
\end{equation*}
where  $z(t)=[u(t-1:t-N)^\mathrm{T} ~ \hat{y}(t-1:t-N)^\mathrm{T} ]^\mathrm{T},$
and $M=[M_u ~ M_y]$,
$M_y=A^{N}V_N^{\dagger}$, $M_u=U_N-M_yT_N$.
\end{lemma}

A clear limitation of Lemma \ref{le:output} is that it relies on prior knowledge of the system matrices $A, B, C$. However, this requirement is relaxed in the following subsections.

\subsection{Data-Driven Estimation of Infinite-Horizon Average Output Energy}

Denote $J_{\infty}(\mathbb{S}_{\gamma})=\textrm{tr}(\mathbb{S}_{\gamma}CW_cC^{\mathrm{T}}\mathbb{S}_{\gamma}^\mathrm{T})$ as the infinite-horizon average output energy \eqref{eq:sensor_problem_trinf}, which we want to estimate for a fixed sensor set $\mathbb{S_{\gamma}}$ using input-output data. The following Lemma re-expresses this cost function with respect to a dual ALE, which will subsequently facilitate the task of its data-driven estimation.

\begin{lemma}\label{le:trace}
It holds that
\begin{equation}\label{eq:J2}
J_{\infty}(\mathbb{S}_{\gamma})=\mathrm{tr}(B^\mathrm{T} W_o^\gamma B),
\end{equation}
where $W_o^\gamma$ uniquely solves the ALE
\begin{equation}\label{eq:disc_gram}
a^2A^\mathrm{T}W_o^\gamma A-W_o^\gamma+C^{\mathrm{T}}\mathbb{S}_{\gamma}^\mathrm{T}\mathbb{S}_{\gamma}C=0.
\end{equation}
\begin{proof}
 The proof follows similar lines to \cite{bernstein2009matrix}. \frQED 
\end{proof}
\end{lemma}

Exploiting Lemma \ref{le:trace} and inspired from reinforcement learning methods \cite{lewis2010reinforcement}, the following lemma provides a data-driven procedure for estimating the infinite-horizon average output energy \eqref{eq:J2} for a fixed set of sensors.  The basis of this data-driven procedure lies in the temporal difference equations associated with the cost  \eqref{eq:H2a_int}.

\begin{lemma}\label{le:data_inf}
Consider system \eqref{eq:sys2} and let Assumption \ref{ass:obs} hold. Then, it holds that
\begin{equation}\label{eq:data1}
J_{\infty}(\mathbb{S}_{\gamma})=\mathrm{tr}\left(E_1^\mathrm{T}\bar{W}_o^\gamma E_1 \right),
\end{equation}
where $E_1=[I_m ~ 0_{m\times (Nr+(N-1)m)} ]^\mathrm{T}$, and
$\bar{W}_o^\gamma$ is a symmetric matrix that satisfies, for all $t\ge N$,
\begin{equation}\label{eq:data2}
{\Phi}^\mathrm{T}(t)\mathrm{vech}(\bar{W}_o^\gamma )+\norm{\tilde{y}(t)}^2=0,
\end{equation}
with
\begin{equation}\label{eq:data3}
\begin{split}
\Phi(t)=&H\Big(a^2((z(t+1)-E_1u(t))\\&\otimes (z(t+1)-E_1u(t)))-z(t)\otimes z(t)\Big)\Big),
\end{split}
\end{equation}
where $z(\cdot)$ was defined in Lemma \ref{le:output}. 
\end{lemma}
\begin{proof}
Multiplying both sides of equation \eqref{eq:disc_gram} of Lemma \ref{le:trace} with $x(t)$, $t\ge N$, and using \eqref{eq:sys2}, we obtain
\begin{equation}\label{eq:th1}
a^2x^\mathrm{T}(t)A^\mathrm{T}W_o^\gamma Ax(t)-x^\mathrm{T}(t)W_o^\gamma x(t)+\norm{\tilde{y}(t)}^2=0.
\end{equation}
Moreover, from \eqref{eq:sys2}, one has $Ax(t)=x(t+1)-Bu(t)$. Substituting this relation in \eqref{eq:th1} yields
\begin{multline}\label{eq:th2}
a^2(x(t+1)-Bu(t))^\mathrm{T}W_o^\gamma (x(t+1)-Bu(t))\\-x^\mathrm{T}(t)W_o^\gamma x(t)+\norm{\tilde{y}(t)}^2=0.
\end{multline}
Employing Lemma \ref{le:output}, we have $x(t+1)=Mz(t+1)$ and $x(t)=Mz(t)$. Moreover, notice that $B=ME_1$. Hence, \eqref{eq:th2} becomes
\begin{multline}\nonumber
a^2(Mz(t+1)-ME_1u(t))^\mathrm{T}W_o^\gamma (Mz(t+1)-ME_1u(t))\\-z^\mathrm{T}(t)M^\mathrm{T}W_o^\gamma Mz(t)+\norm{\tilde{y}(t)}^2=0.
\end{multline}
Defining the matrix $\bar{W}_o^\gamma=M^\mathrm{T}W_o^\gamma M$ yields
\begin{multline}\label{eq:th3}
a^2(z(t+1)-E_1u(t))^\mathrm{T}\bar{W}_o^\gamma (z(t+1)-E_1u(t))\\-z^\mathrm{T}(t)\bar{W}_o^\gamma z(t)+\norm{\tilde{y}(t)}^2=0.
\end{multline}
Note now that
\begin{align*}
&(z(t+1)-E_1u(t))^\mathrm{T}\bar{W}_o^\gamma (z(t+1)-E_1u(t))\\&=\Big((z(t+1){-}E_1u(t)){\otimes} (z(t+1){-}E_1u(t))\Big)^\mathrm{T}\mathrm{vec}(\bar{W}_o^\gamma),\\
&z^\mathrm{T}(t)\bar{W}_o^\gamma z(t)=(z(t)\otimes z(t))^\mathrm{T}\textrm{vec}(\bar{W}_o^\gamma ).
\end{align*}
Using these relations, we linearly parameterize \eqref{eq:th3} as
\begin{equation*}
\bar{\Phi}^\mathrm{T}(t)\textrm{vec}(\bar{W}_o^\gamma )+\norm{\tilde{y}(t)}^2=0,
\end{equation*}
where $\bar{\Phi}(t)=a^2\Big((z(t+1)-E_1u(t))\otimes (z(t+1)-E_1u(t))\Big)-z(t)\otimes z(t)$. Applying half-vectorization thanks to the symmetricity of $\bar{W}_o^\gamma$ gives equations \eqref{eq:data2}-\eqref{eq:data3}. Finally, recall that $\bar{W}_o^\gamma=M^\mathrm{T}{W}_o^\gamma M$, and $ME_1=B$. Therefore, it follows that $\textrm{tr}(E_1^\mathrm{T}\bar{W}_o^\gamma E_1)=\textrm{tr}(E_1^\mathrm{T}M^\mathrm{T}{W}_o^\gamma M E_1)=\textrm{tr}(B^\mathrm{T}{W}_o^\gamma B)=J_{\infty}(\mathbb{S}_\gamma)$, 
with the last equality following from Lemma \ref{le:trace}. This yields equation \eqref{eq:data1}, and concludes the proof. \frQED
\end{proof}

Note now that equation \eqref{eq:data2} is linearly parameterized with respect to $\bar{W}_o^\gamma$. As a result, one can solve it for $\bar{W}_o^\gamma$ with a least-squares procedure, provided a number of sufficiently rich input-output measurements is available. This claim is summarized in the following Corollary.

\begin{corollary}\label{cor:1}
Let $t_0,t_1,\ldots,t_k$ be measurement time instants, and let Assumption \ref{ass:obs} hold. Denote
\begin{align*}
\Psi&:=\begin{bmatrix}\Phi(t_0) & \Phi(t_1) & \ldots & \Phi(t_k) \end{bmatrix},\\
Y_{\gamma}&:=\begin{bmatrix}\norm{\tilde{y}(t_0)}^2 & \norm{\tilde{y}(t_1)}^2 & \ldots & \norm{\tilde{y}(t_k)}^2 \end{bmatrix}^\mathrm{T}.
\end{align*}
If $\Psi$ has full row rank, then
\begin{equation}\label{eq:Wo_sol}
\mathrm{vech}(\bar{W}_o^\gamma )=-(\Psi^\mathrm{T})^\dagger Y_{\gamma}.
\end{equation}
\end{corollary}
\begin{proof}
Stacking equations of the form \eqref{eq:data2} for $t=t_1,\ldots,t_k$ in a vertical matrix, we obtain
\begin{equation}\label{eq:bf_ls}
\Psi^\textrm{T}\mathrm{vech}(\bar{W}_o^\gamma )=-Y_{\gamma}.
\end{equation}
Hence, the result follows since $\Psi$ has full row rank. \frQED
\end{proof}
Combining the assumptions and conditions of Lemma \ref{le:data_inf} and Corollary \ref{cor:1}, we finally obtain a complete data-driven expression for the cost function $J_{\infty}(\mathbb{S}_{\gamma})$.
\begin{corollary}\label{cor:2}
Let Assumption \ref{ass:obs} hold and $\Psi$ have full row rank.  Then:
\begin{equation}\label{eq:Jdata}
J_{\infty}(\mathbb{S}_{\gamma})=-\mathrm{tr}\left(E_1^\mathrm{T}\mathrm{vech}^{-1}\left( ((\Psi^\mathrm{T})^\dagger Y_{\gamma}\right) E_1 \right). \frqed
\end{equation}

\end{corollary}

It is important to note that Corollary \ref{cor:1} builds on the assumption that $\Psi\in\mathbb{R}^{(Nm+Nr)(Nm+Nr+1)/2\times k}$  has full row rank.  While this property is attainable when $Nr=n$, it is not when $Nr>n$, where $N$ is the number of past output data, $r$ the size of the output vector $\hat{y}$, and $n$ the size of the state vector. The underlying reason is that the system's output 
cannot provide more information than the information contained in the system's state. Therefore, when the total dimension $Nr$ of the past output data exceeds the dimension $n$ of the system state, the past output data turns linearly dependent and $\Psi$ becomes rank deficient. Accordingly,  the maximal rank $\Psi$ can achieve is $(Nm+n)(Nm+n+1)/2$, even though its minimal dimension is $(Nm+Nr)(Nm+Nr+1)/2$.

The aforementioned detail can be problematic in practice, particularly when the fraction $\frac{n}{r}$ is not an integer number, or when $N$ is chosen to be conservatively large. The main issue is that, in these cases, equation \eqref{eq:bf_ls} has infinitely many solutions, and only one of them is equal to the matrix of interest $\bar{W}_o^\gamma$. Nonetheless, despite the fact that formula \eqref{eq:Wo_sol} may not be valid in extracting $\bar{W}_o^\gamma$ in these cases, the following theorem shows that it still allows one to accurately estimate the infinite-horizon sensor selection metric $J_{\infty}(\mathbb{S}_{\gamma})$, for a fixed set of sensors.

\begin{theorem}\label{th:2}
Let Assumption \ref{ass:obs} hold and  $\hat{W}_o^\gamma$ satisfy
\begin{equation}\label{eq:Wo_sol_pinv}
\mathrm{vech}(\hat{W}_o^\gamma )=-(\Psi^\mathrm{T})^\dagger Y_{\gamma}.
\end{equation}
If  $\mathrm{rank}({\Psi})=(Nm+n)(Nm+n+1)/2$, then 
\begin{equation}\label{eq:costWhat}
J_{\infty}(\mathbb{S}_{\gamma})=\mathrm{tr}\left(E_1^\mathrm{T}\hat{W}_o^\gamma E_1 \right).
\end{equation}
\end{theorem}
\begin{proof}
Since $\hat{W}_o^\gamma$ satisfies \eqref{eq:Wo_sol_pinv}, it follows that $\Psi^\textrm{T}\mathrm{vech}(\hat{W}_o^\gamma )=-Y_{\gamma}.$
Devectorizing this equation and de-stacking it for all time instances, one can bring it back to the form \eqref{eq:th3}, i.e., the following equation holds for all $t=t_1,\ldots,t_k$:
\begin{multline}\label{eq:pinv1}
a^2(z(t+1)-E_1u(t))^\mathrm{T}\hat{W}_o^\gamma (z(t+1)-E_1u(t))\\-z^\mathrm{T}(t)\hat{W}_o^\gamma z(t)+\norm{\tilde{y}(t)}^2=0.
\end{multline}
Define the error $\tilde{W}_o^\gamma=\bar{W}_o^\gamma-\hat{W}_o^\gamma$. Subtracting \eqref{eq:pinv1} from \eqref{eq:th3}, we obtain for all $t=t_1,\ldots,t_k$:
\begin{multline}\label{eq:pinv2}
a^2(z(t+1)-E_1u(t))^\mathrm{T}\tilde{W}_o^\gamma (z(t+1)-E_1u(t))\\-z^\mathrm{T}(t)\tilde{W}_o^\gamma z(t)=0.
\end{multline}
Note now that by using the telescopic sum on \eqref{eq:sys}, the list of outputs from $t-1$ to $t-N$ can be expressed as
\begin{equation*}
\hat{y}(t-1:t-N)=V_Nx(t-N)+\bar{T}_Nu(t-2:t-N),
\end{equation*}
where $\bar{T}_N$ is equal to $T_N$ with the first $m$ (zero) columns removed, and
where the matrices $V_N, ~T_N$ were defined in Lemma \ref{le:output}. Therefore, if we define
\begin{align}\label{eq:zL}
\hspace{-3mm}\bar{z}(t){=}\begin{bmatrix}u(t-1) \\ u(t-2:t-N) \\ x(t-N) \end{bmatrix},
\Lambda{=} \begin{bmatrix}I_m & 0 & 0 \\ 0 & I_{(N-1)m} & 0 \\ 0 & \bar{T}_N & V_N \end{bmatrix},\hspace{-2mm}
\end{align}
 we can write $z(t)=\Lambda \bar{z}(t)$. Hence, \eqref{eq:pinv2} turns into
\begin{multline}\label{eq:pinv3}
a^2(\Lambda \bar{z}(t+1)-{E}_1u(t))^\mathrm{T}\tilde{W}_o^\gamma (\bar{z}(t+1)-E_1u(t))\\-\bar{z}^\mathrm{T}(t)\Lambda^\textrm{T}(t)\tilde{W}_o^\gamma\Lambda \bar{z}(t)=0.
\end{multline}
In addition, note that we have $E_1=\Lambda\bar{E}_1$, where 
\begin{equation}\label{eq:barE1}
\hspace{-3mm}\bar{E}_1{=}
\begin{bmatrix}I_m \\ 0_{((N-1)m+n)\times m}   \end{bmatrix}.\hspace{-1mm}
 \end{equation}
 Using this, \eqref{eq:pinv3} can be written for all $t=t_1,\ldots,t_k$ as
\begin{multline}\label{eq:pinv4}
a^2( \bar{z}(t+1)-\bar{E}_1u(t))^\mathrm{T}\Lambda^\textrm{T}\tilde{W}_o^\gamma\Lambda (\bar{z}(t+1)-\bar{E}_1u(t))\\-\bar{z}^\mathrm{T}(t)\Lambda^\textrm{T}\tilde{W}_o^\gamma\Lambda \bar{z}(t)=0.
\end{multline}
Vectorizing this equation we obtain  $\tilde{\Phi}^\textrm{T}(t)\textrm{vech}(\Lambda^\textrm{T} \tilde{W}_o^\gamma \Lambda)\allowbreak=0$ for all $t=t_1,\ldots,t_k$,
where
\begin{multline}\label{eq:tPhi}
\tilde{\Phi}(t)=H\Big(a^2((\bar{z}(t+1)-\bar{E}_1u(t))\\\otimes (\bar{z}(t+1)-\bar{E}_1u(t)))-\bar{z}(t)\otimes \bar{z}(t)\Big).
\end{multline}
Defining $\tilde{\Psi}:=\begin{bmatrix}\tilde{\Phi}(t_0) & \tilde{\Phi}(t_1) & \ldots & \tilde{\Phi}(t_k) \end{bmatrix}$, we obtain
\begin{equation}\label{eq:barPsi}
\tilde{\Psi}^\textrm{T}\textrm{vech}(\Lambda^\textrm{T} \tilde{W}_o^\gamma \Lambda)=0.
\end{equation}
At this point, it is useful to notice from \eqref{eq:data3} and \eqref{eq:tPhi}, and from the fact that $z(t)=\Lambda\bar{z}(t)$ and $E_1=\Lambda\bar{E}_1$, that
\begin{equation*}
Q_2{\Phi}(t)=(\Lambda\otimes\Lambda)  Q_1\tilde{\Phi}(t),
\end{equation*}
where $Q_1\in\mathbb{R}^{(Nm+n)^2\times(Nm+n)(Nm+n+1)/2}$
 and $Q_2\in\mathbb{R}^{(Nm+Nr)^2\times(Nm+Nr)(Nm+Nr+1)/2}$ are
the matrix expressions for the (full column rank) linear operator $H^{-1}(\cdot)$  on the corresponding vector spaces. Accordingly,
\begin{equation}\label{eq:tildePsi}
Q_2{\Psi}=(\Lambda\otimes\Lambda)  Q_1 \tilde{\Psi}.
\end{equation}
Since $(A,\hat{C})$ is observable, $V_N$ has full column rank, hence $\Lambda$ has full column rank, and finally so does $(\Lambda \otimes \Lambda)$. Moreover, since $Q_1$ has full column rank, it follows that $(\Lambda \otimes \Lambda)Q_1$ also has full column rank, so that $\rank ((\Lambda\otimes\Lambda)  Q_1)=(Nm+n)(Nm+n+1)/2$. In addition, $\textrm{rank}(\Psi)=(Nm+n)(Nm+n+1)/2$ by assumption, and thus also $\textrm{rank}(Q_2\Psi)=(Nm+n)(Nm+n+1)/2$. Thus, by applying the rank inequality in \eqref{eq:tildePsi}, we get
\begin{multline}\nonumber
(Nm+n)(Nm+n+1)/2=\textrm{rank}(Q_2\Psi)\\\le\min\{\textrm{rank}((\Lambda \otimes \Lambda)Q_1),\textrm{rank}(\tilde{\Psi}) \}\\=\min\{(Nm+n)(Nm+n+1)/2,\textrm{rank}(\tilde{\Psi}) \}.
\end{multline}
This implies $\textrm{rank}(\tilde{\Psi})\ge (Nm+n)(Nm+n+1)/2$, and  thus $\textrm{rank}(\tilde{\Psi})=(Nm+n)(Nm+n+1)/2$ since $\tilde{\Psi}\in\mathbb{R}^{(Nm+n)(Nm+n+1)/2\times k}$. Hence $\tilde{\Psi}$ has full row rank, hence  \eqref{eq:barPsi} has the unique solution $\textrm{vech}(\Lambda^\textrm{T} \tilde{W}_o^\gamma \Lambda)=0$, hence $\Lambda^\textrm{T} \tilde{W}_o^\gamma \Lambda=0.$
 Accordingly, one has $\bar{E}_1^\textrm{T}\Lambda^\textrm{T} \tilde{W}_o^\gamma \Lambda\bar{E}_1=0$, and thus $E_1^\textrm{T}\tilde{W}_o^\gamma {E}_1=0$ since $E_1=\Lambda \bar{E}_1$. Therefore, we conclude
$E_1^\textrm{T}\hat{W}_o^\gamma {E}_1=E_1^\textrm{T}\bar{W}_o^\gamma {E}_1, $
and thus \eqref{eq:costWhat} follows from Lemma \ref{le:data_inf} and \eqref{eq:data1}. \frQED
\end{proof}

Thanks to Theorem \ref{th:2}, we conclude that the sensor selection cost $J_{\infty}(\mathbb{S}_\gamma)$ can be computed via the pseudoinverse of $\Psi$, even if $\Psi$ does not have full row rank\footnote{Having a data matrix $\Psi$ with full row rank may still be desirable as it can enable a faster computation of $\Psi^\dagger$.}. This is still subject to a (less restrictive) rank condition on $\Psi$, but this condition -- unlike the full row rank one -- is attainable in practice when $Nr>n$.

\subsection{Data-Driven Estimation of Finite-Horizon Average Output Energy}

While infinite-horizon metrics like \eqref{eq:trinf} have the advantage of being computable from static ALEs, they may not be applicable in cases where the system is deployed over a fixed time interval, or when a sufficiently small discount factor $a$ is not a priori known. In such cases, finite-horizon indices like \eqref{eq:trfin} are more desirable to optimize, and in this section we propose a procedure for their data-driven estimation.

Denote as $J_{T}(\mathbb{S}_{\gamma})=\textrm{tr}(\mathbb{S}_{\gamma}CW_c(T)C^{\mathrm{T}}\mathbb{S}_{\gamma}^\mathrm{T})$ the finite-horizon average output energy \eqref{eq:sensor_problem_trfin}, which we want to estimate or a fixed sensor set $\mathbb{S_{\gamma}}$ using input-output data. The following Lemma will facilitate this goal and is stated without proof as it follows from \cite{bernstein2009matrix}. 

\begin{lemma}\label{le:model_fin}
It holds that
\begin{equation}\label{eq:J2modelfin}
J_{T}(\mathbb{S}_{\gamma})=\mathrm{tr}(B^\mathrm{T} W_o^\gamma(T) B),
\end{equation}
where $W_o^\gamma(T)$ is obtained by propagating the following DLE over $t\in\{0,
\ldots,T-1\}$
\begin{equation}\label{eq:gram_fin}
A^\mathrm{T}W_o^\gamma(t)A+C^{\mathrm{T}}\mathbb{S}_{\gamma}^\mathrm{T}\mathbb{S}_{\gamma}C=W_o^\gamma(t+1),~W_o^\gamma(0)=0.
\end{equation}
\end{lemma}
Note now that the DLE \eqref{eq:gram_fin} depends on the system matrices $A$, $C$, while a direct dependence on $B$ is present in the trace expression \eqref{eq:J2modelfin}. The following Lemma recasts \eqref{eq:gram_fin} into a data-driven, time-varying Lyapunov-like equation, which can be propagated without knowledge of these unknown matrices. In addition, the data used in that direction do not need to have the same timestamp as the corresponding time variable in \eqref{eq:gram_fin}.
\begin{lemma}\label{le:data_fin}
Consider system \eqref{eq:sys2}, let Assumption \ref{ass:obs} hold, and let $T\in\mathbb{N}$. Then, it follows that
\begin{equation}\label{eq:J2fin}
J_{T}(\mathbb{S}_{\gamma})=\mathrm{tr}(E_1^\mathrm{T} \bar{W}_o^\gamma(T) E_1),
\end{equation}
where $E_1=[I_m ~ 0_{m\times (Nr+(N-1)m)}]^\mathrm{T}$, and $\bar{W}_o^\gamma(t)$ is a symmetric matrix for all $t\in\{1,\ldots,T\}$, which satisfies the following data-driven dynamics
\begin{multline}\label{eq:data_LEt}
\Phi_1^\mathrm{T}(\tau)\mathrm{vech}(\bar{W}_o^\gamma(t+1)))\\=\Phi_2^\mathrm{T}(\tau)\mathrm{vech}(\bar{W}_o^\gamma(t)))+\norm{\tilde{y}(\tau)}^2,~\bar{W}_o^\gamma(0)=0,
\end{multline}
where $\Phi_1(\tau)=H(z(\tau)\otimes z(\tau))$ and $\Phi_2(\tau)=H((z(\tau+1)-E_1u(\tau))\otimes (z(\tau+1)-E_1u(\tau)))$, $\tau\ge N$.
\end{lemma}
\begin{proof}
Let $\tau\ge N$. Multiplying \eqref{eq:gram_fin} of Lemma \ref{le:model_fin} from the left with $x^\textrm{T}(\tau)$ and from the right with $x(\tau)$, we get:
\begin{equation}\nonumber
x^\textrm{T}(\tau)W_o^\gamma(t+1)x(\tau)=x^\textrm{T}(\tau)A^\mathrm{T}W_o^\gamma(t)Ax(\tau)+\norm{\tilde{y}(\tau)}^2.
\end{equation}
Since $Ax(\tau)=x(\tau+1)-Bu(\tau)$, this property yields:
\begin{multline}\nonumber
x^\textrm{T}(\tau)W_o^\gamma(t+1)x(\tau)=(x(\tau+1)-Bu(\tau))^\textrm{T}\\\cdot W_o^\gamma(t)(x(\tau+1)-Bu(\tau))+\norm{\tilde{y}(\tau)}^2{.}
\end{multline}
Using Lemma \ref{le:output}, we have $x(\tau)=Mz(\tau)$, where $ME_1=B$. Therefore:
\begin{multline}\label{eq:Wot1}
z^\textrm{T}(\tau)M^\textrm{T}W_o^\gamma(t+1)Mz(\tau)\\=(z(\tau+1)-E_1u(\tau))^\textrm{T}M^\textrm{T}W_o^\gamma(t)\\\cdot M(z(\tau+1)-E_1u(\tau))+\norm{\tilde{y}(\tau)}^2.
\end{multline}
Defining $\bar{W}_o^\gamma(t)=M^\textrm{T}W_o^\gamma(t)M$, \eqref{eq:Wot1} becomes
\begin{multline}\label{eq:Wot2}
z^\textrm{T}(\tau)\bar{W}_o^\gamma(t+1)z(\tau)=(z(\tau+1)-E_1u(\tau))^\textrm{T}\bar{W}_o^\gamma(t)\\\cdot (z(\tau+1)-E_1u(\tau))+\norm{\tilde{y}(\tau)}^2.
\end{multline}
Note now that 
\begin{align}\nonumber
&z^\textrm{T}(\tau)\bar{W}_o^\gamma(t+1)z(\tau)\\&\nonumber\quad=H(z(\tau)\otimes z(\tau))^\textrm{T}\textrm{vech}(\bar{W}_o^\gamma(t+1))),\\\nonumber
&(z(\tau+1)-E_1u(\tau))^\textrm{T}\bar{W}_o^\gamma(t)(z(\tau+1)-E_1u(\tau))\\&\nonumber\quad=H((z(\tau+1)-E_1u(\tau))\otimes (z(\tau+1)-E_1u(\tau)))^\textrm{T}\\&\qquad\qquad\qquad\qquad\qquad\qquad\qquad\cdot\textrm{vech}(\bar{W}_o^\gamma(t))).\label{eq:Wot3}
\end{align}
Combining \eqref{eq:Wot2}-\eqref{eq:Wot3} yields \eqref{eq:data_LEt}.  Moreover, since $\bar{W}_o^\gamma(t)=M^\mathrm{T}{W}_o^\gamma(t) M$ for all $t\in\{0,\ldots,T\}$, and since $ME_1=B$, from Lemma \ref{le:model_fin} we get $\textrm{tr}(E_1^\mathrm{T}\bar{W}_o^\gamma(T) E_1)=\textrm{tr}(E_1^\mathrm{T}M^\mathrm{T}{W}_o^\gamma(T) M E_1){=}\textrm{tr}(B^\mathrm{T}{W}_o^\gamma(T) B){=}J_{T}(\mathbb{S}_\gamma)$. \hspace{-1mm}\frQED
\end{proof}

Similar to any data-driven scheme, the Lyapunov-like equation \eqref{eq:data_LEt} can be propagated to obtain $\bar{W}_o^\gamma(T)$ only if sufficient input-output data have been collected. We summarize this requirement in the following Corollary.
\begin{corollary}\label{cor:fullrow2}
Let $t_0,t_1,\ldots,t_k$, be measurement time instants with respect to $\tau\ge N$, and let Assumption \ref{ass:obs} hold. Denote
\begin{align*}
\Psi_1&:=\begin{bmatrix}\Phi_1(t_0) & \Phi_1(t_1) & \ldots & \Phi_1(t_k) \end{bmatrix},\\
\Psi_2&:=\begin{bmatrix}\Phi_2(t_0) & \Phi_2(t_1) & \ldots & \Phi_2(t_k) \end{bmatrix},\\
Y_{\gamma}&=\begin{bmatrix}\norm{\tilde{y}(t_0)}^2 & \norm{\tilde{y}(t_1)}^2 & \ldots & \norm{\tilde{y}(t_k)}^2 \end{bmatrix}^\mathrm{T}.
\end{align*}
If $\Psi_1$ has full row rank, then the Lyapunov-like equation \eqref{eq:data_LEt} can be exactly propagated at each step $t\in\{0,\ldots,T-1\}$ according to the least-squares law:
\begin{align}\nonumber
\mathrm{vech}(\bar{W}_o^\gamma(t+1) )&=(\Psi_1^\mathrm{T})^\dagger (\Psi_2^\mathrm{T}\mathrm{vech}(\bar{W}_o^\gamma(t) )+Y_{\gamma}),\\
\mathrm{vech}(\bar{W}_o^\gamma(0) )&=0.\label{eq:Wot_sol}
\end{align}
\end{corollary}
\begin{proof}
Stacking the data-driven dynamics \eqref{eq:data_LEt} in a vertical matrix for all $\tau\in\{t_0,t_1,\ldots,t_k\}$, we obtain
\begin{equation}\label{eq:tmp}
\Psi_1^\mathrm{T}\mathrm{vech}(\bar{W}_o^\gamma(t+1) )= \Psi_2^\mathrm{T}\mathrm{vech}(\bar{W}_o^\gamma(t) )+Y_{\gamma}.
\end{equation}
Since $\Psi_1$ has full row rank, equation \eqref{eq:Wot_sol} follows. \frQED
\end{proof}
As a final remark, we note that when the observability matrix $V_N$ (defined in Lemma \ref{le:output}) is redundant, in the sense that it has rank $n$ but more than $n$ rows, the matrix $\bar{W}_o^\gamma(T)$ cannot be exactly extracted. However, in what follows, we prove that obtaining any solution to \eqref{eq:tmp} is sufficient to estimate the metric \eqref{eq:J2fin}.
\begin{theorem}\label{th:4}
Let Assumption \ref{ass:obs} hold, and let $\hat{W}_o^\gamma(t)$ be obtained, for all $t\in\{0,\ldots,T-1\}$ from
\begin{align}\nonumber
\mathrm{vech}(\hat{W}_o^\gamma(t+1) )&=(\Psi_1^\mathrm{T})^\dagger (\Psi_2^\mathrm{T}\mathrm{vech}(\hat{W}_o^\gamma(t) )+Y_{\gamma}),\\
\mathrm{vech}(\hat{W}_o^\gamma(0) )&=0.\label{eq:Wohatt_sol}
\end{align}
If  $\mathrm{rank}(\Psi_1)=(Nm+n)(Nm+n+1)/2$, then
\begin{equation}\label{eq:costWhat_t}
J_{T}(\mathbb{S}_{\gamma})=\mathrm{tr}\left(E_1^\mathrm{T}\hat{W}_o^\gamma(T) E_1 \right).
\end{equation}
\begin{proof}
We will use induction for this proof.

1) At $t=0$, we have $\bar{W}_o^\gamma(0)=M^\textrm{T}{W}_o^\gamma(0)M=0$, and $\hat{W}_o^\gamma(0)=0$ by \eqref{eq:Wohatt_sol}. Therefore, $\bar{W}_o^\gamma(0)=\hat{W}_o^\gamma(0)$ and, accordingly, $\Lambda^\textrm{T}\bar{W}_o^\gamma(0)\Lambda=\Lambda^\textrm{T}\hat{W}_o^\gamma(0)\Lambda$, with $\Lambda$ as in \eqref{eq:zL}.

2) Assume that at time $t$, we have $\Lambda^\textrm{T}\bar{W}_o^\gamma(t)\Lambda=\Lambda^\textrm{T}\hat{W}_o^\gamma(t)\Lambda$. We want to prove that $\Lambda^\textrm{T}\bar{W}_o^\gamma(t+1)\Lambda=\Lambda^\textrm{T}\hat{W}_o^\gamma(t+1)\Lambda$ under the considered assumptions. To that end, note that by defining $\tilde{W}_o^{\gamma}(t)=\bar{W}_o^{\gamma}(t)-\hat{W}_o^{\gamma}(t)$, if we subtract \eqref{eq:Wohatt_sol} from \eqref{eq:tmp} we get
\begin{equation}\nonumber
\Psi_1^\mathrm{T}\mathrm{vech}(\tilde{W}_o^\gamma(t+1) )=\Psi_2^\mathrm{T}\mathrm{vech}(\tilde{W}_o^\gamma(t)).
\end{equation}
Devectorizing this equation yields, for $\tau\in\{t_0,\ldots,\allowbreak t_k\}$:
\begin{multline}\label{eq:tmp2}
z^\textrm{T}(\tau)\tilde{W}_o^\gamma(t+1)z(\tau)=(z(\tau+1)-E_1u(\tau))^\textrm{T}\\\cdot \tilde{W}_o^\gamma(t)(z(\tau+1)-E_1u(\tau)).
\end{multline}
Following the line of proof of Theorem \ref{th:2} it follows that $z(\tau)=\Lambda \bar{z}(\tau)$, and $E_1=\Lambda\bar{E}_1$, where $\bar{z}(\tau),~\Lambda,~\bar{E}_1$ are as in \eqref{eq:zL}, \eqref{eq:barE1}. Hence, \eqref{eq:tmp2} turns into
\begin{multline}\label{eq:tmp3}
\bar{z}^\textrm{T}(\tau)\Lambda^\textrm{T}\tilde{W}_o^\gamma(t+1)\Lambda\bar{z}(\tau)=(\bar{z}(\tau+1)-\bar{E}_1u(\tau))^\textrm{T}\\\cdot\Lambda^\textrm{T}\tilde{W}_o^\gamma(t)\Lambda (\bar{z}(\tau+1)-\bar{E}_1u(\tau)).
\end{multline}
From the inductive assumption we have $\Lambda^\textrm{T}\bar{W}_o^\gamma(t)\Lambda=\Lambda^\textrm{T}\hat{W}_o^\gamma(t)\Lambda$, and therefore $\Lambda^\textrm{T}\tilde{W}_o^\gamma(t)\Lambda=0$. Hence the right-hand side of \eqref{eq:tmp3} vanishes, and \eqref{eq:tmp3} turns into:
\begin{equation*}
\bar{z}^\textrm{T}(\tau)\Lambda^\textrm{T}\tilde{W}_o^\gamma(t+1)\Lambda\bar{z}(\tau)=0.
\end{equation*}
Vectorizing this equation, we get 
\begin{equation*}
\tilde{\Phi}_1^\mathrm{T}(\tau)\mathrm{vech}(\Lambda^\textrm{T}\bar{W}_o^\gamma(t+1))\Lambda)=0
\end{equation*}
for all $\tau\in\{t_0,\ldots,t_k\}$, where $\tilde{\Phi}_1(t)=H\Big(\bar{z}(t)\otimes \bar{z}(t)\Big)$. Accordingly, if we define $\tilde{\Psi}_1:=[\tilde{\Phi}_1(t_0) ~ \tilde{\Phi}_1(t_1) ~ \ldots ~ \tilde{\Phi}_1(t_k)]$, 
we obtain $\tilde{\Psi}_1^\textrm{T}\textrm{vech}(\Lambda^\textrm{T} \tilde{W}_o^\gamma(t+1) \Lambda)=0$. 
Given the assumption that $\mathrm{rank}(\Psi_1)=(Nm+n)(Nm+n+1)/2$, and following identically the proof of Theorem \ref{th:2} below equation \eqref{eq:barPsi}, it then follows that $\Lambda^\textrm{T} \tilde{W}_o^\gamma(t+1) \Lambda=0$, which concludes the induction. Hence, we have $\Lambda^\textrm{T} \tilde{W}_o^\gamma(T) \Lambda=0$, and also $\bar{E}_1^\textrm{T}\Lambda^\textrm{T} \tilde{W}_o^\gamma(T) \Lambda\bar{E}_1=0$. Since $E_1=\Lambda \bar{E}_1$, this implies that $E_1^\textrm{T}\tilde{W}_o^\gamma(T) {E}_1=0$, and thus  $E_1^\textrm{T}\hat{W}_o^\gamma(T) {E}_1=E_1^\textrm{T}\bar{W}_o^\gamma(T) {E}_1$. 
The final result \eqref{eq:costWhat_t} then follows from Lemma \ref{le:data_fin}. \frQED
\end{proof} 
\end{theorem}

\section{Data-Driven Sensor Selection}\label{sec:data_select}

\subsection{Data-Driven Greedy Algorithm}

The previous section evaluated, in a data-driven manner, the $\mathcal{H}_2$ norm that a given sensor set yields over an infinite or a finite horizon. The next step is to use this approach to find the best set of sensors without knowledge of the system. Towards this end, the following lemma provides an unsurprising result, given the modularity of the $\mathcal{H}_2$ norm proved in \cite{summers2014optimal}. That is, to evaluate $J_T(\mathbb{S}_{\gamma})$ (or $J_{\infty}(\mathbb{S}_{\gamma})$), the lemma shows that one only needs to evaluate $J_T(\mathbb{S}_{\gamma_i})$ (or $J_{\infty}(\mathbb{S}_{\gamma_i})$) for all $i=1,\ldots,p'$.

\begin{lemma}\label{le:selection}
It holds that $J_T(\mathbb{S}_{\gamma})=\sum_{i=1}^{p'}J_T(\mathbb{S}_{\gamma_i}),~T\in\mathbb{N}$, and
$J_{\infty}(\mathbb{S}_{\gamma})=\sum_{i=1}^{p'}J_{\infty}(\mathbb{S}_{\gamma_i}).$
\end{lemma}
\begin{proof}
For any $T\in\mathbb{N}$, we have the relation that $J_T(\mathbb{S}_{\gamma}){=}\textrm{tr}(\mathbb{S}_{\gamma}CW_c(T)C^\mathrm{T}\mathbb{S}_{\gamma}^\mathrm{T}){=}\sum_{i=1}^{p'}\textrm{tr}(c_{\gamma_i}W_c(T)c_{\gamma_i}^\textrm{T})=\sum_{i=1}^{p'}J_T(\mathbb{S}_{\gamma_i}).$
The proof for $J_{\infty}(\mathbb{S}_{\gamma})$  is identical. \frQED
\end{proof}

The implication of Lemma \ref{le:selection} is that, to solve the sensor selection problem, one does not need to perform data-driven estimation of $J_T(\mathbb{S}_{\gamma})$ (or $J_{\infty}(\mathbb{S}_{\gamma})$) for all possible selection matrices $\mathbb{S}_{\gamma}$; rather, one only needs to evaluate the cost $J_T(\mathbb{S}_{j})$, $j=1,\ldots,p$, of each distinct sensor row in $C$. Subsequently, sorting the scores $J_T(\mathbb{S}_{j})$ for all $j=1,\ldots,p$, and choosing $\mathbb{S}_{\gamma^\star}$ to contain the $p'$ sensors with the highest scores indeed solves the sensor selection problem \eqref{eq:sensor_problem_trfin} (or \eqref{eq:sensor_problem_trinf}). This procedure is summarized in Algorithm \ref{al:placement}, and its convergence properties stated formally in the following Theorem.

\begin{algorithm}[!t]
\caption{Data-Driven Sensor Selection}
\hspace*{\algorithmicindent} \textbf{Input}: Horizon $T\in\mathbb{N}\cup\{\infty\}$.
\begin{algorithmic}[1]
\Procedure{}{}
\For {$j=1,\ldots,p$}
\State $\tilde{C}\leftarrow \mathbb{S}_{j}C$.
\State Gather input-output data from \eqref{eq:sys2}.
\State Evaluate $J_T(\mathbb{S}_j)$ from \eqref{eq:costWhat} or \eqref{eq:costWhat_t}.
\EndFor
\State Sort $J_T(\mathbb{S}_j)$ for $j=1,\ldots,p$ in decreasing order, \hspace*{4mm} let $\gamma^\star=\{\gamma_1^\star,\ldots,\gamma_{p'}^\star\}$ contain the $p'$ indices with \hspace*{4mm} the highest scores.
\State Select $C'\leftarrow \mathbb{S}_{\gamma^\star}C$.
\EndProcedure
\end{algorithmic}\label{al:placement}
\end{algorithm}

\begin{theorem}
\begin{enumerate}
\item Suppose that $T=\infty$. Then, Algorithm \ref{al:placement} converges to the optimal sensor selection matrix \eqref{eq:sensor_problem_trinf} under the Assumptions of Theorem \ref{th:2}.
\item Suppose that $T\in\mathbb{N}$. Then, Algorithm \ref{al:placement} converges to the optimal sensor selection matrix \eqref{eq:sensor_problem_trfin} under the Assumptions of Theorem \ref{th:4}.
\end{enumerate}
\begin{proof}
Follows from Theorems \ref{th:2}, \ref{th:4} and Lemma \ref{le:selection}. \frQED
\end{proof}
\end{theorem}

\begin{remark}
While Algorithm \ref{al:placement} sequentially evaluates the sensor selection metrics using different input-output data, this evaluation can also take place using the same data for all $j$. However, this requires having the resources to employ all to-be-evaluated sensors at once. \frqed
\end{remark}

\subsection{Computational Complexity}

The most expensive operation in Algorithm \ref{al:placement} is the computation of the pseudoinverse $\Psi^\dagger$ in the infinite-horizon case, and $\Psi_1^\dagger$ in the finite-horizon one, whose complexity is $\mathcal{O}(k^3)$. Given that these pseudoinverses have been calculated, the greedy algorithm in the infinite-horizon case then needs to perform $p$ operations of the form \eqref{eq:Wo_sol_pinv}, with complexity $\mathcal{O}(pk^2)$. On the other hand, the greedy algorithm in the finite-horizon case needs to perform $p$ operations of the form \eqref{eq:Wohatt_sol} up to $T$ times, with complexity  $\mathcal{O}(pTk^2)$. Hence, the overall computational complexity in the infinite-horizon case is $\mathcal{O}(k^3+pk^2)$, while in the finite-horizon case it is $\mathcal{O}(k^3+pTk^2)$.

\subsection{Optimization of Volumetric Observability Measures}\label{sec:log}

Throughout this work,  the focus has been on choosing sensors that optimize the $\mathcal{H}_2$ norm. However, one may also consider volumetric observability measures captured by the following log determinants \cite{ manohar2021optimal}:
\begin{equation}\label{eq:detinf}
\begin{split}
J'_{\infty}(\mathbb{S}_{\gamma})&=\textrm{logdet}(B^\mathrm{T}W_o^{\gamma} B),\\
J'_T(\mathbb{S}_{\gamma})&=\textrm{logdet}(B^\mathrm{T}W_o^{\gamma}(T)B),
\end{split}
\end{equation}
where $W_o^{\gamma}$ and $W_o^{\gamma}(T)$ are the infinite and finite horizon observability Gramians from \eqref{eq:disc_gram} and \eqref{eq:gram_fin}. Notably, although the $\mathcal{H}_2$ norm that was studied in the previous sections involved a trace instead of a log determinant, that trace's arguments were the same as the arguments in \eqref{eq:detinf}. As a result, we immediately get the following result for the data-driven estimation of  \eqref{eq:detinf}. 
\begin{corollary}\label{cor:data_det}
\begin{enumerate}
    \item Let the Assumptions of Theorem \ref{th:2} hold. Then $J'_{\infty}(\mathbb{S}_{\gamma})=\mathrm{logdet}(E_1^\mathrm{T}\hat{W}_o^\gamma E_1 ),$
    where $\hat{W}_o^\gamma$ is obtained from \eqref{eq:Wo_sol_pinv}.
    \item Let the Assumptions of Theorem \ref{th:4} hold. Then $J'_{T}(\mathbb{S}_{\gamma})=\mathrm{logdet}(E_1^\mathrm{T}\hat{W}_o^\gamma(T) E_1 ),$
    where $\hat{W}_o^\gamma(T)$ is obtained by propagating \eqref{eq:Wohatt_sol}.
\end{enumerate}    
\end{corollary}
\begin{proof}
The proof follows those of Theorems \ref{th:2}, \ref{th:4}. \frQED
\end{proof}
Unlike the $\mathcal{H}_2$ norm, volumetric measures are not additive functions.  Nevertheless, they still fall within a class of well-behaved set functions called \textit{submodular}, which can be approximately optimized by a greedy algorithm with a guaranteed optimality ratio of at least $1-e^{-1}$ \cite{summers2015submodularity, summers2014optimal}.  In view of Corollary \ref{cor:data_det}, it is straightforward to extend such a greedy algorithm to the data-driven case.

\section{Data-Driven Verification of Observability}\label{sec:obs}

A known issue with greedy selection is that it can lead to a resulting set of sensors that is not observable \cite{summers_structural}. In that regard, modifications can be incorporated into the greedy algorithm to enforce it to converge to a fully observable set, particularly by encouraging it to increase the rank of an observability test at each iteration \cite{summers2016convex}. However, performing such an observability test becomes unclear when the system matrices are unknown and when access to only partial state measurements from the system is available.

Building on the results of the previous sections, in what follows, we provide a data-driven condition for verifying whether a given set of sensors $\tilde{C}$ is observable.

\begin{theorem}\label{th:obs}
Let  $t_0,t_1,\ldots,t_k\in\mathbb{N}$, $k\ge Nm+n$, and consider the data matrices
$Z=[z(t_0) ~ z(t_1) ~ \ldots ~ z(t_k)],
\tilde{Z}=[\tilde{z}(t_0) ~ \tilde{z}(t_1) ~ \ldots ~ \tilde{z}(t_k)]$,
where $z(t)$ is defined in Lemma \ref{le:output} and $\tilde{z}(t)=[u(t-1:t-N)^\mathrm{T} ~ \tilde{y}(t-1:t-N)^\mathrm{T}]^\mathrm{T}$.
\begin{enumerate}
\item If $\mathrm{rank}(\tilde{Z})=Nm+n$ then $(A,\tilde{C})$ is observable.
\item Let $\mathrm{rank}(Z)=Nm+n$. Then, $(A,\tilde{C})$ is observable if and only if $\mathrm{rank}(\tilde{Z})=Nm+n$.
\end{enumerate}
\end{theorem}
\vspace{-2mm}
\begin{proof}
To prove the first item, from \eqref{eq:sys}, the outputs from $t-1$ to $t-N$ can be written as $\hat{y}(t-1:t-N)=\tilde{V}_Nx(t-N)+\tilde{T}_Nu(t-1:t-N),$
where $\tilde{V}_N, ~\tilde{T}_N$ are as in Lemma \ref{le:output} but with $\hat{C}$ substituted by $\tilde{C}$. Then, defining
\begin{align}\nonumber
\bar{z}(t)=\begin{bmatrix}u(t-1:t-N) \\ x(t-N) \end{bmatrix},~
\tilde{\Lambda}= \begin{bmatrix}I_{Nm} & 0  \\  \tilde{T}_N & \tilde{V}_N \end{bmatrix},\hspace{-2mm}
\end{align}
and $\bar{Z}=[\bar{z}(t_0)~\ldots~\bar{z}(t_k)]$, it follows that $\tilde{Z}=\tilde{\Lambda}\bar{Z}$. Hence, since $\mathrm{rank}(\tilde{Z})=Nm+n$, we obtain 
\begin{equation}\label{eq:LZ1rank}
\mathrm{rank}(\tilde{\Lambda}\bar{Z})=Nm+n.
\end{equation}
At the same time, $\bar{Z}\in\mathbb{R}^{ (Nm+n)\times k}$, and thus $\textrm{rank}(\bar{Z})\le Nm+n$. Given this, \eqref{eq:LZ1rank} can hold only if 
\begin{equation}\label{eq:Lrank1}
\mathrm{rank}(\tilde{\Lambda})\ge Nm+n.
\end{equation}
 Next, note that the first $Nm$ columns of $\tilde{\Lambda}$ are linearly independent of the last $n$ columns. Therefore, 
\begin{equation}\label{eq:Lrank2}
Nm+\textrm{rank}(\tilde{V}_N)\ge \mathrm{rank}(\tilde{\Lambda}).
\end{equation}
Combining \eqref{eq:Lrank1}-\eqref{eq:Lrank2} we obtain $\textrm{rank}(\tilde{V}_N)\ge n$, and hence $\textrm{rank}(\tilde{V}_N)=n$, meaning $(A,~\tilde{C})$ is observable. 

To prove the second item, recall from the proof of Theorem \ref{th:2} that $z(t)=\Lambda \bar{z}(t)$.
Hence, we have $Z=\Lambda \bar{Z}$, and since $\textrm{rank}(Z)=Nm+n$, it follows that $\textrm{rank}(\bar{Z})=Nm+n$. We now show necessity and sufficiency.
\begin{itemize}
\item Suppose that $(A, \tilde{C})$ is observable. Then, $\textrm{rank}(\tilde{\Lambda})=Nm+n$, and thus $\textrm{rank}(\tilde{Z})=Nm+n$ since $\tilde{Z}=\tilde{\Lambda}\bar{Z}$ and $\tilde{\Lambda}, \bar{Z}$ are full rank matrices with rank $Nm+n$.
\item  Suppose that $\textrm{rank}(\tilde{Z})=Nm+n$. Then, from the first item of the theorem, $(A, \tilde{C})$ is observable. \frQED
\end{itemize}

\end{proof}

\begin{remark}
In the presence of noise, it should be checked if the minimum singular value of $\tilde{Z}$ is above an empirical threshold, instead of checking its rank.\frqed
\end{remark}

Using Theorem \ref{th:obs}, there are several ways in which observability of the selected set of sensors could be enforced in the previous section. One approach is to add heuristics to the greedy algorithm that promote full-rank observability, as in \cite{summers2016convex}. Another is to relax the cardinality constraint in the greedy selection algorithm, selecting additional sensors until observability of the resulting set is verified using Theorem \ref{th:obs}.  Finally, initializing the algorithm with an observable set (e.g.,  $\hat{C}$) ensures that the final selection remains observable \cite{summers_structural}.

\section{Simulations}\label{sec:sim}

\subsection{Power System}

To showcase how the proposed approach can learn and optimize infinite-horizon observability metrics, we carry out simulations on the 10-machine New England power system. Following \cite{vafaee2024learning}, the dynamics of this system are given by $m_i\ddot{\theta}_i+d_i\dot{\theta}_i=-\sum_{j\sim i}k_{ij}(\theta_i-\theta_j)+u_i$, $i=1,\ldots,10$, where $\theta_i~[\textrm{rad}]$ represents the rotor angle of the generator at bus $i$, $\dot{\theta}_i$ its corresponding frequency, $m_i~[\textrm{pu-sec2/rad}]$ its inertia, $d_i~[\textrm{pu-sec/rad}]$ its damping coefficient, $u_i$ its mechanical power input,  and $k_{ij}$ the line suspectance for connected generators indicated by $i\sim j$. All (unknown) parameter values are taken from \cite{atawi2013advance}. This system can be written in the compact form:
\begin{align*}
\frac{\textrm{d}}{\textrm{d}t}\begin{bmatrix}\theta \\ \dot{\theta} \end{bmatrix}&=\begin{bmatrix}0 & I_{10} \\ -M^{-1}L & -M^{-1}D \end{bmatrix}\begin{bmatrix}\theta \\ \dot{\theta} \end{bmatrix}+\begin{bmatrix}0\\ M^{-1} \end{bmatrix}u(t),\\y(t)&=C_c\begin{bmatrix}\theta^\mathrm{T} & \dot{\theta}^\mathrm{T} \end{bmatrix}^\mathrm{T},
\end{align*}
where $\theta=[\theta_1~\ldots~\theta_{10}]^\textrm{T}$, $\dot{\theta}=[\dot{\theta}_1~\ldots~\dot{\theta}_{10}]^\textrm{T}$, $M=\textrm{diag}(m_1,\ldots,m_{10})$, $D=\textrm{diag}(d_1,\ldots,d_{10})$, $u=[u_1~\ldots\allowbreak~u_{10}]^\textrm{T}$, and $L$ is the Laplacian matrix. In addition, $C_c$ represents the set of available sensors, each of which can measure a distinct state of the system. The system is subsequently discretized with a timestep of $0.2~\textrm{sec}$, to obtain the discrete-time system matrices $A,B,C$ \cite{fazelnia2016convex}.

We assume the initial sensor set is $\gamma_0=\hat{\gamma}=\{1,3,4,6,8\}$,  hence the matrix $\hat{C}$ of Assumption \ref{ass:obs} consists of rows $1,3,4,6,8$ of the matrix $C$. To increase the observability of the system, we aim to add $7$ more sensors to $\hat{C}$, out of the remaining set of sensors $\{2,5,7,9,10,11,12,13,14,15,16,17,18,19,20\}$. To this end, we employ a data-driven greedy algorithm, gathering input-output data from the system under the control input proposed in \cite{lewis2010reinforcement} augmented with exploration noise, and aiming to learn the sensors that optimize the infinite-horizon volumetric measure in \eqref{eq:detinf} without knowledge of the system. We choose the discount factor as $a=0.99$, and $N=4$. Overall, $2000$ data samples were collected, and the greedy algorithm required $2.53$ seconds to run, out of which $1.39$ seconds were spent to calculate the pseudoinverse of $\Psi$. By contrast, a data-driven brute-force search required $62.13$ seconds.

\begin{figure}[!t] 
		\centering
                \includegraphics[width=6.8cm,height=5cm]{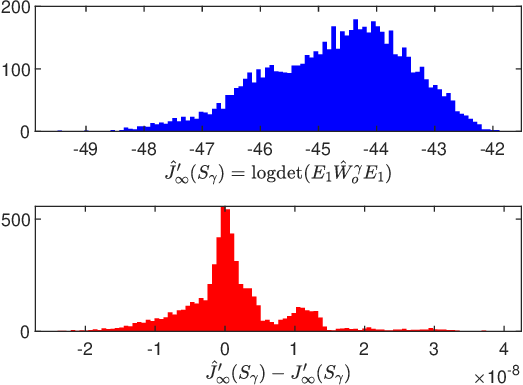}
                \caption{ Histogram of a)  the estimation $J'_{\infty}(\mathbb{S}_{\gamma})$ of $\hat{J}'_{\infty}(\mathbb{S}_{\gamma})$ for all $\gamma$ with cardinality $7$, using the data-driven formula of Corollary \ref{cor:data_det}; and c) the corresponding estimation errors.}
                \label{fig:hist}
\end{figure}

 \begin{center}
\begin{table}[!t]
  \begin{center}
  \setlength{\tabcolsep}{4pt} % Reduce column separation
\renewcommand{\arraystretch}{1.0} % Reduce row spacing
\small
  \caption{Performance of chosen sensor set across different metrics.\vspace{3mm} \label{tab:perform}}
  \resizebox{\columnwidth}{!}{%
   \begin{tabular}{c||c|c|c|c|c}
      \toprule % <-- Toprule here
      \textbf{ } & $\mathrm{logdet}(B^\mathrm{T}W_o^\gamma B)$ & $\mathrm{tr}(B^\mathrm{T}W_o^\gamma B)$ & $\mathrm{tr}(W_o^\gamma)$ & $\mathrm{tr}({W_o^\gamma}^{-1})$ & $\lambda_{\textrm{min}}(W_o^\gamma)$\\
      \midrule % <-- Midrule here
      Initial  & $-57.4$ & $8.9\cdot10^{-2}$ & $88.4$ & $2.8\cdot10^{4}$ & $4.4\cdot10^{-5}$\\
      Final & $-41.9$ & $0.2$ & $198.3$ & $25.9$ & $0.1$ \\
      \bottomrule % <-- Bottomrule here
    \end{tabular}}
  \end{center}
\end{table}
\end{center}

Figure \ref{fig:hist} shows the histogram of the estimated cost $J'_{\infty}(\mathbb{S}_{\gamma})=\mathrm{logdet}(E_1^\mathrm{T}\hat{W}_o^\gamma E_1 )$ provided by the proposed data-driven procedure according to Corollary \ref{cor:data_det}, as well as the histogram of the associated estimation errors. We verify that the data-driven estimation of the cost is indeed accurate, with the estimation errors being in the range of $10^{-8}$. In addition, the estimated approximately optimal set of sensors provided by the greedy algorithm  is $\gamma^\star=\{1,3,4,6,8\}\cup\{2,5,7,9,10,13,18\}$, which is exactly equal to the optimal set of sensors, with a cost function equal to $-41.89$. Table \ref{tab:perform}  compares the performance of the initial and final sensor sets across five standard metrics. We observe a consistent improvement after optimization, confirming that the chosen metric aligns well with the system's observability structure.

\begin{figure}[t]
    \centering
    \subfloat[Error in estimating $J_T(S_j)$\\ for all $j$ and $T$.\label{fig:JT}]{
        \includegraphics[width=0.53\columnwidth]{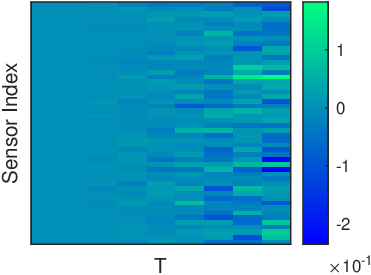}
    }
    \hfill
    \subfloat[Optimal sensor selection for all $T$.\label{fig:schedule}]{
        \includegraphics[width=0.42\columnwidth]{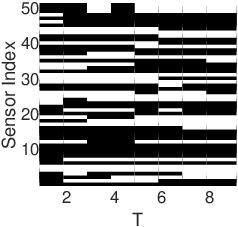}
    }
    \caption{Comparison between estimation error and optimal sensor selection over horizon lengths $T=1,\dots,8$.}
    \label{fig:sidebyside_23}
\end{figure}

\begin{figure}[!t] 
		\centering
                \includegraphics[width=6.2cm,height=4.5cm]{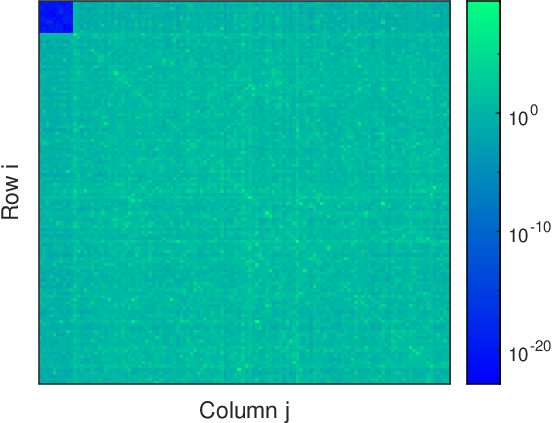}
                \caption{The mean squared value of the relative error $[\hat{W}_o^\gamma(T)-\bar{W}_o^{\gamma}(T)]_{ij}/[\bar{W}_o^{\gamma}(T)]_{ij}$ for each entry $i,j$ and for $T=8$, over all sensor indices.}
                \label{fig:Wtilde}
\end{figure}

\subsection{Large Scale System}

To further showcase how the proposed approach can learn and optimize \textit{finite-horizon} observability metrics, we perform simulations on a random system with $n=50$ states, $m=10$ actuators, and $p=50$ available sensors. The initial set of sensors consists of $r=10$ units, and the objective is to substitute them with $p'=20$ units that maximize the finite-horizon observability metric \eqref{eq:trfin} with $T=8$, without knowledge of the system's dynamics. To that end, we employ Algorithm \ref{al:placement}, gathering input-output data from the system, with $N=6$, and using them to approximate the optimal sensor set.  Overall, $6700$ data samples were collected, and the greedy algorithm required $78.29$ seconds to run, out of which $68.59$ seconds were spent to calculate the pseudoinverse of $\Psi_1$.

According to Lemma \ref{le:selection}, we have $J_T(\mathbb{S}_{\gamma})=\sum_{i=1}^{20'}J_T(\mathbb{S}_{\gamma_i})$ for any $\gamma$ with a cardinality of $20$, so the problem boils down to accurately estimating the costs $J_T(\mathbb{S}_{j})$ for all $j=1,\ldots,50$ using the data-driven formula \eqref{eq:costWhat_t}. Figure \ref{fig:JT} shows the resulting errors in estimating these costs, over all values of $T=1,\ldots,8$. From this figure, it is clear that the errors are practically negligible, and in the order of $10^{-12}$. Using Algorithm \ref{al:placement}, Figure \ref{fig:schedule} also shows the optimal choice of sensors for all values of $T=1,\ldots,8$. Finally, Figure \ref{fig:Wtilde} validates Theorem \ref{th:4}, in which we proved that the upper-left part of the matrix $\bar{W}_o^\gamma(T)$ can be uniquely identified from input-output data. Note that it is this part that is needed to calculate $J_T$.

\section{Conclusion}\label{sec:conc}

We studied input-output data-driven sensor selection for CPS with unknown system models. We developed a model-free algorithm to estimate metrics of observability and select sensors using input-output data. 

Future work includes extending our framework to online settings, following \cite{fotiadis2021learning}. Another open problem is that of polynomial-time data-driven sensor selection with strict observability guarantees. While details on enforcing observability were provided in Section \ref{sec:obs}, they relied on relaxing the scope of the problem.

%\balance

\newcommand{\BIBdecl}{\setlength{\itemsep}{-0.15 em}}
\bibliographystyle{IEEEtran}       
\bibliography{Placement_bib}

\end{document}